\documentclass[journal]{IEEEtran}
\usepackage{graphicx}
\usepackage{subfigure}
\usepackage{epsfig}
\usepackage{amsfonts}
\usepackage{amssymb}
\usepackage{amsmath}
\usepackage{caption}

\newtheorem{theorem}{Theorem}

\newtheorem{condition}[theorem]{Condition}
\newtheorem{conjecture}[theorem]{Conjecture}
\newtheorem{corollary}[theorem]{Corollary}

\newtheorem{definition}[theorem]{Definition}

\input{tcilatex}
\begin{document}

\title{Diameter Perfect Lee Codes}
\author{ Peter Horak, Bader F. AlBdaiwi }
\maketitle

\begin{abstract}
Lee codes have been intensively studied for more than 40 years. Interest in
these codes has been triggered by the Golomb-Welch conjecture on the
existence of the perfect error-correcting Lee codes. In this paper we deal
with the existence and enumeration of diameter perfect Lee codes. As main
results we determine all $q$ for which there exists a linear diameter-$4$
perfect Lee code of word length $n$ over $Z_{q},$ and prove that for each $%
n\geq 3$ there are uncountable many diameter-$4$ perfect Lee codes of word
length $n$ over $Z.$ This is in a strict contrast with perfect
error-correcting Lee codes of word length $n$ over $Z\,$\ as there is a
unique such code for $n=3,$ and its is conjectured that this is always the
case when $2n+1$ is a prime. We produce diameter perfect Lee codes by an
algebraic construction that is based on a group homomorphism. This will
allow us to design an efficient algorithm for their decoding. We hope that
this construction will turn out to be useful far beyond the scope of this
paper.
\end{abstract}


\let\thefootnote\relax\footnotetext{\newline
\textit{{\ Peter Horak is with the Interdisciplinary Arts and Sciences,
University of Washington, Tacoma, 1900 Commerce St., Tacoma, WA 98402-3100,
USA (email: horak@uw.edu).\newline
Bader F. AlBdaiwi is with the Computer Science Department, Kuwait
University, P. O. Box 5969, Safat, 13060, Kuwait (email:
bader.albdaiwi@ku.edu.kw).\newline
} }}


\section{\protect\bigskip Introduction}

\noindent The most common metric in coding theory is the Hamming metric. In
this paper we deal with another frequently used metric, so called Lee
metric. This metric was introduced in \cite{Lee} and \cite{Ulrich} for
transmission of signals. The main motive power for studying codes in Lee
metric goes back to the seminal paper of Golomb and Welch \cite{GW}. In this
paper we will study the existence and will enumerate diameter perfect Lee
codes.\bigskip

\noindent First we recall some definitions and notation. As usual, let $Z$
be the set of integers, $Z_{q}$ denote the integers modulo $q,$ and let $%
T^{n}$ be the $n$-fold Cartesian product of a set $T.$ Throughout the paper
we will use $Z_{q}^{n}$ and $Z^{n}$ for the Abelian (component-wise)
additive group on $Z_{q}^{n}$ and $Z^{n}$ as well. It will always be clear
from the context whether we have in mind a set or the group on this set.\
Because of the coding theory context the elements of $Z_{q}^{n}$ and $Z^{n}$
will also be called words. The Lee distance (=the Manhattan distance, the
zig-zag distance) $\rho _{L}(v,w)$ of two words $v=(v_{1},v_{2},...,v_{n}),$ 
$w=(w_{1},...,w_{n})$ is given by $\rho _{L}(v,w)=\sum\limits_{i=1}^{n}\min
(\left\vert v_{i}-w_{i}\right\vert ,q-\left\vert v_{i}-w_{i}\right\vert )$
for $u,v\in \mathbb{Z}_{q}^{n},$ and by $\rho _{L}(v,w)=$ $%
\sum\limits_{i=1}^{n}\left\vert v_{i}-w_{i}\right\vert $ for $v,w\in Z^{n}.$
By $S_{n,r}(v)$ we denote the Lee sphere of radius $r$ centered at $v;$ that
is, $S_{n,r}(v)=\{w;\rho _{L}(v,w)\leq r\}\,$. {\ For }$v,w$ with $\rho
_{L}(v,w)=1,$ {the double Lee sphere $DS_{n,r}(v,w)$ of radius }$r$ with its
center comprising $v$,$w$ {is the set $S_{n,r}(v)\cup S_{n,r}(w)\,$}.\bigskip

\noindent Perfect error-correcting Lee codes can be introduced in several
ways. For the purpose of this paper it is convenient to define them through
tilings. In order not to have to formulate the same statement/definition
twice, once for $Z^{n}$ and then for $Z_{q}^{n},$ the symbol $Z_{\ast }^{n}$ 
$\ $stands for both of them; that is, by $Z_{\ast }^{n}$ we mean $Z^{n}$ or $%
Z_{q}^{n}.$\ We set $O=(0,...,0)$ and $e_{i}=(0,...,0,1,0,...,0),$ where the 
$i$-th coordinate equals $1.$\bigskip

\noindent Let $V$ be a subset of $Z_{\ast }^{n}$. \ By a copy of $V$ we mean
an image of $V$ under a linear distance preserving bijection on $Z_{\ast
}^{n}.$ Translations of $V$ form a special family of copies of $V.$ We
recall that by a translation of $V$ we understand the set $V+x=\{w;w=x+v,$ $%
v\in V\}$ where $x\in Z_{\ast }^{n}$. As the sphere $S_{n,r}$ is symmetric,
each copy of the sphere is its translation. Each copy of the double-sphere $%
DS_{n,r}$ can be obtained by a translation of the double-sphere $%
DS_{n,r}(O,e_{i})$ for some $i,1\leq i\leq n.$ A collection $\mathcal{T=\{}%
V_{i};i\in I\}$ of copies of $V$ constitutes a tiling of $Z_{\ast }^{n}$ by $%
V$ if $\bigcup\limits_{i\in I}V_{i}=Z_{\ast }^{n}$ and $V_{i}\cap
V_{j}=\emptyset $ for all $i\neq j.$ A tiling $\mathcal{T}$ of $Z_{\ast
}^{n} $ by $V$ that consists of translations of $V$ can be described as $%
\mathcal{T=\{}V+l;l\in \mathcal{L}\}$ where $\mathcal{L\subset }Z_{\ast }^{n}
$. $\ $A tiling $\mathcal{T}$ by translations of $V$ is called periodic if $%
\mathcal{L}$ is periodic, and $\mathcal{T}$ is called a lattice tiling if $%
\mathcal{L}$ is a lattice (=a subgroup) of $Z_{\ast }^{n}$.\bigskip

\noindent Let $V=S_{n,r}$ in $Z_{\ast }^{n}$ such that, if $Z_{\ast
}^{n}=Z_{q}^{n}$ then $q\geq 2r+1.$ A set $\mathcal{L}\subset Z_{\ast }^{n}$
is called a perfect $r$-error-correcting Lee code if $\mathcal{T=\{}V+l;l\in 
\mathcal{L\}}$ is a tiling of $Z_{\ast }^{n}.$ If $\mathcal{L}$ is a lattice
then $\mathcal{L}$ is called a linear code. For $Z_{\ast }^{n}=Z^{n}$ the
perfect $r$-error-correcting Lee code is denoted by $PL(n,r)$ while for $%
Z_{\ast }^{n}=Z_{q}^{n}$ it is denoted by $PL(n,r,q).$ The long-standing
Golomb-Welch conjecture deals with the existence of $PL(n,r)$ codes.
Although there is a vast literature on the subject the conjecture is still
far from being solved.\bigskip

\noindent In this paper we focus on diameter-$d$ perfect Lee codes, which
constitute a generalization of perfect error-correcting Lee codes. Ahlswede
et al., see \cite{A1}, introduced diameter perfect codes for distance
regular graphs. ~Let $(M,\rho )$ be a metric space. Then a set $C\subset M$
is a diameter-$d$ code if $\rho (u,v)\geq d$ for any $u,v\in C,$ and a set $%
A\subset M$ is an anticode of diameter $d$ if $\rho (u,v)\leq d$ for all $%
u,v\in A.$ Each graph $G$ can be seen as a metric space. It is proved in 
\cite{A1} that if $G$ is a distance regular graph, $C$ is a diameter-$d$
code in $G,$ and $A$ is an anticode of diameter $d-1$ of the maximum size in 
$G$ then $\left\vert C\right\vert \left\vert A\right\vert \leq \left\vert
G\right\vert ;$ that is, the sphere packing bound applies to distance
regular graphs as well. A diameter-$d$ code that attains the bound in the
above inequality is called diameter-$d$ perfect code. This definition can be
used for example in Hamming scheme as the related graph is distance regular.
Etzion \cite{E} showed that the sphere packing bound applies to $Z_{q}^{n}$
endowed with Lee metric as well although the related graph is not distance
regular; analogically to \cite{A1} he defined the diameter-$d$ perfect Lee
code as a diameter-$d$ code that attains the equality in the sphere packing
bound. Clearly, a definition through the sphere packing bound cannot be
extended to an infinite space. To deal with the case of $Z^{n}$ Etzion
proved that if $\mathcal{L}$ is a lattice in $Z^{n}$ with the minimum
distance $d$ and a volume $\left\vert V(\mathcal{L)}\right\vert $ then $%
\left\vert A\right\vert \leq \left\vert V(\mathcal{L)}\right\vert $ for any
anticode of diameter $d-1.$ Then he called $\mathcal{L}$ to be the diameter
perfect code if $\mathcal{L}$ attains the equality in the bound. As far as
we know there is no formal definition of the diameter perfect Lee code in $%
Z^{n}$ in case when $\mathcal{L}$ is not a lattice. \bigskip

\noindent Now we provide a unified definition of a diameter-$d$ perfect Lee
code for both spaces $Z_{q}^{n}$ and $Z^{n}$ regardless whether the code is
a linear one. Let $\mathcal{S}=\{S_{i};i\in I\}$ be a family of subsets of
an underlying set $M.$ Then a set $T\subset M$ is called a transversal of $%
\mathcal{S}$ if there is a bijection $f:I\rightarrow T$ so that $f(i)\in
S_{i}.$ Equivalently, $T$ is a transversal of $\mathcal{S}$ if $\left\vert
T\cap S_{i}\right\vert =1$ for each $i\in I,$ and $T\cap S_{i}\neq T\cap
S_{j}$ for all $i\neq j\in I.$\bigskip

\begin{definition}
Let $\mathcal{L}\subset Z_{\ast }^{n}$. Then $\mathcal{L}$ is a diameter-$d$
perfect Lee code in $Z_{\ast }^{n}$ if $\mathcal{L}$ is a diameter-$d$ code,
and there is a tiling $\mathcal{T=}\{W_{i};i\mathcal{\in }I\mathcal{\}}$ of $%
Z_{\ast }^{n}$ by the anticode of diameter $d-1$ of maximum size such that $%
\mathcal{L}$ is a transversal of $\mathcal{T}.$ The diameter-$d$ perfect Lee
code in $Z^{n}$ will be denoted by $DPL(n,d),$ and in $Z_{q}^{n}$ by $%
DPL(n,d,q)$.\bigskip
\end{definition}

\noindent In other words, a diameter-$d$ code $\mathcal{L}$ is perfect if
there is a tiling $\mathcal{T}$ by anticodes of diameter $d-1$ of the
maximum size so that each tile in $\mathcal{T}$ contains exactly one
codeword of $\mathcal{L}$. It is not difficult to see that for $Z_{q}^{n}$
and for lattice case in $Z^{n}$ our definition of diameter perfect Lee codes
is equivalent to the definition due to Etzion.\bigskip

\noindent Any error-correcting perfect Lee code is also a diameter perfect
Lee code. Indeed, it is easy to see that, for $d$ even, the anticode of
diameter $d\,$ of the maximum size is the Lee sphere $S_{n,r}$ with $r=\frac{%
d}{2}.$ Thus, for $d$ odd, $DPL(n,d)$ and $DPL(n,d,q)$ codes are $PL(n,r)$
and $PL(n,r,q)$ codes where $r=\frac{d-1}{2}$, respectively. It was proved
in \cite{A} that, for $d$ odd, the anticode of diameter $d$ of maximum size
is the double-sphere $DS_{n,r}$ with $r=\frac{d-1}{2}.$ \bigskip

\noindent We point out an advantage of our definition. Suppose that we need
to show the existence of a diameter-$d$ perfect Lee code in $Z_{\ast
}^{n}\,\ $where $d$ is even. Then all we need to do is to prove that there
is a tiling $\mathcal{T}=\{W_{i};i\in I\}$ of $Z_{\ast }^{n}$ by the
double-sphere $DS_{n,r}$ of radius $r=\frac{d-2}{2}$.\ Indeed, choose a
fixed tile $W_{i_{0}}$ in $\mathcal{T}.$ Let $x$ be one of the two words
forming the center of $W_{i_{0}.\text{ }}$We will prove that the set $N%
\mathcal{=\{}y;y$ belongs to the center of $W_{i}$ and $\rho _{L}(x,y)$ is
even, $i\in I\}$ is a distance-$d$ perfect Lee code. Clearly, $N$ is well
defined as the center of $W_{i}$ consists of two words at distance $1;$ that
is, one of them at even and the other at odd distance from $x.$ $N$ is a
transversal of $\mathcal{T},$ thus we need only to prove that $\rho
_{L}(x,y)\geq d$ for any $x,y\in N$. Clearly, the distance of any two words
in $N$ is even. This is obvious to see for $Z_{\ast }^{n}=Z^{n}.\,\ $\ It is
not true in general for $Z_{\ast }^{n}=Z_{q}^{n},$ it is valid only in the
case when $q$ is even. However, as the size of $DS_{n,r}$ is even (see \cite%
{E}) $|Z_{q}^{n}|=q^{n}$ has to be even as well, which in turn implies $q$
is even. So for the distance of any two words in $N$ we have $\rho
_{L}(x,y)\geq r+r+1;$ and since $\rho _{L}(x,y)$ is even, we get $\rho
_{L}(x,y)\geq 2r+2=d.$ We have proved:\bigskip

\begin{theorem}
\label{X}If there is a tiling of $Z_{\ast }^{n}$ by double-spheres $DS_{n,r}$
then there is a distance-$d$ perfect Lee code in $Z_{\ast }^{n}$ with $%
d=2r+2.\ $Moreover, if $\{W+l;l\in \mathcal{L}\}$ is a tiling of $Z_{\ast
}^{n}$ by translations of $DS_{n,r}$ such that the weight $\left\vert
l\right\vert =\rho _{L}(l,O)$ of $l$ is even for all $l\in \mathcal{L}$ then 
$\mathcal{L}$ is a diameter-$d$ perfect Lee code.\bigskip
\end{theorem}

\noindent The main goal of this paper is to study the existence and to
enumerate diameter perfect Lee codes. It turns out that some results on
diameter perfect Lee codes and on perfect error-correcting Lee codes are
similar but in some cases the results on the two types of codes are very
much different. E.g., it was proved in \cite{HB} that there is only one, up
to isomorphism, tiling of $Z^{3}$ by Lee spheres $S_{3,1}$ but we show in
this paper that there are uncountable many tilings of $Z^{3}$ by double Lee
sphere $DS_{3,1}$. The rest of the paper is organized as follows. \
Golomb-Welch conjecture claims that there are no $PL(n,r)$ codes for $n\geq
3,$ and $r>1.$ The current state of this conjecture as well as its extension
to $DPL(n,d)$ codes and a further extension to perfect dominating sets is
provided in Section 2. We defined perfect Lee codes by means of tilings. In 
\cite{S1} S. Stein introduced an algebraic construction of lattice tilings 
based on group homomorphisms. He used this construction for tilings  by
crosses. This construction, and its variations occur in many papers, see
e.g. \cite{ADH, S3, HS, HB, Molnar, Sw,  S2,  S, Sz}. In Section 3 a
generalization of Stein's construction is given. Let $V$ be a subset of $%
Z^{n}.$ We state a necessary and sufficient condition for the existence of a
lattice tiling of $Z^{n}$ by $V$ in terms of a group homomorphism. We guess
that this construction will turn useful outside the scope of this
paper.\bigskip 

\noindent The main results of this paper are stated in Section 4. It is
believed, see Section 2, that, with only one exception of $DPL(3,6)$ code,
there are no $DPL(n,d)$ codes, and also no $DPL(n,d,q)$ codes for $n\geq 3,$
and $d>4.$ All values of $q$ for which there exists a linear $PL(n,1,q)$
code (=a linear $DPL(n,3,q)$ code) were determined in \cite{H}. We prove an
analogous result for linear $DPL(n,4,q)$ codes. As to the enumeration of all 
$DPL(n,4)$ codes, Etzion \cite{E} proved that there are uncountable many of
them if $n$ is a power of $2.$ We show that this is also the case for the
other values of $n\geq 3.$ The situation with $PL(n,1)$ codes is very much
different. It was proved in \cite{HB} that there are uncountable many $%
PL(n,1)$ codes if $2n+1$ is not a prime. However, on the other hand it was
showed there that there is only one, up to isomorphism, $PL(3,1)$ code, and
the authors conjectured that the same is true for each $n$ where $2n+1$ is a
prime. We note that the existence of uncountable many $DPL(n,4)$ codes
implies that there are uncountable many non-periodic $DPL(n,4)$ codes as the
total number of periodic $DPL(n,4)$ codes is at most $\aleph _{0}.$
Non-periodic $DPL(n,4)$ codes were constructed by means of non-regular
lattice tilings (i.e. tilings that are not a face-to-face tilings) of $R^{n}$
by a cluster of cubes centered at words of a double-sphere $DS_{n,1}.$ This
is another example of a well know situation when to solve a problem we need
to generalize it first. We note that results in Section 4 answer in the
affirmative questions (4), (5), and (8) raised by Etzion in \cite{E}.
Question (6) from \cite{E} was answered in \cite{HB}. \ Finally, in Section
5 an efficient decoding algorithm for linear perfect Lee codes, both
diameter and error-correcting, is designed. Thanks to the representation of
these codes by tilings constructed via a group homomorphism, the computation
complexity of this algorithm is $\Theta (n)$ for $DPL(n,d),d=3,4,$ codes,
and its complexity is $O(\log d)$ for $DPL(2,d)$ codes. We recall that $%
DPL(n,d)$ codes are believed not to exists for $n\geq 3$ and $d>4.$

\section{\protect\bigskip Golomb-Welch conjecture and its extensions}

\noindent In this section we provide a short account of the current state of
Golomb-Welch conjecture and some of its extensions. In \ \cite{GW} it is
conjectured that:\bigskip

\begin{conjecture}
There is no $PL(n,r)$ code for $n\geq 3$ and $r>1.$
\end{conjecture}

\noindent If true then the conjecture is best possible as the existence of $%
PL(n,r)$ codes for $n=2$ and all $r\geq 1\,$, and $n\geq 3$ and $r=1$ was
showed by several authors, see e.g. \cite{GW}. It is also proved in \cite{GW}
that for each $n\geq 3$ there is $r_{n},$ $r_{n}$ not specified, so that
there is no $PL(n,r)$ code with $r\geq r_{n}.$ Although there are plenty of
results in the literature on this conjecture the non-existence of $PL(n,r)$
codes has been proved only in few cases. First, Gravier et al. \cite{Gra}
settled the Golomb-Welch conjecture for $n=3$ and all $r>1.$ Later \v{S}%
pacapan \cite{Sca}, whose proof is computer aided, showed the non-existence
of a $PL(n,r)$ code for $n=4$ and all $r>1$. Horak \cite{Ho} provided an
algebraic proof that there is no $PL(n,r)$ code for $3\leq n\leq 5$ and all $%
r>1.$The only other value of parameters for which the Golomb-Welch
conjecture is known to be true is $n=$ $6$ and $r=2,$ see \cite{Hor}.\bigskip

\noindent The non-existence of $PL(n,r,q)$ codes for some pairs $(n,r)$ and
specific values of $q$ depending on $(n,r)$ was stated in several papers,
see e.g. \cite{AS} for a comprehensive account. The best result of this type
is due to Post \cite{P} who showed that there is \ no $PL(n,r,q)$ code for
any $q\geq 2r+1$ (= code with a large alphabet),$~$\ and $3\leq n\leq
5,r\geq n-2,$ and for $n\geq 6,$ and $r\geq \frac{\sqrt{2}}{2}n-\frac{1}{4}(3%
\sqrt{2}-2).$ Post's result can also be viewed as follows: For the given
pairs of $(n,r)$ there is no periodic $PL(n,r)$ code. In \cite{Hor} and \cite%
{HM} the first author of this paper quoted incorrectly Post's result as he
omitted the word \textit{periodic.\bigskip }

\noindent As mentioned in the Introduction, $DPL(n,d)$ codes are a
generalization of $PL(n,r)$ codes as, for $d$ odd, a $DPL(n,d)$ code is a $%
PL(n,\frac{d-1}{2})$ code as well. Therefore the conjecture stated by Etzion
in \cite{E} is an extension of the Golomb-Welch conjecture.\bigskip

\begin{conjecture}
\cite{E} There is no $DPL(n,d)$ code for $n\geq 3$ and $d>4$ with the
exception of the pair $(n,r)=(3,6).$\bigskip
\end{conjecture}

\noindent Etzion's conjecture, if true, is also best possible. The existence
of $DPL(2,d)$ codes, for $d$ odd was stated in \cite{GW} in terms of $%
PL(n,r) $ codes. The existence of $DPL(n,d)$ codes for $n=2$ and all even $%
d\geq 4,$ and \ all $n\geq 3$ and $d=4$ was established in \cite{E} and \cite%
{ADH}, while the existence of a $DPL(3,6)$ code follows from a Minkowski
tiling \cite{M}.\bigskip

\noindent To be able to present an extension of Etzion's conjecture we state
a generalization of diameter perfect Lee codes introduced in \cite{ADH}. Let 
$G=(V,E)$ be a graph, and $t$ be a natural number. A set $S\subset V$ is
called a $t$-\textit{perfect distance-dominating set} in $G$, or a $t$-%
\textrm{PDDS} in $G$, if, for each $v\in V$, there is a unique component $%
C_{v}$ of $[S]$ so that $d(v,C_{v})\leq t,~$\ and there is in $C_{v}$ a
unique vertex $w$ with $d(v,w)=d(v,C_{v}).$ We recall that the distance $%
d(v,C)$ of a vertex $v\in V$ to $C\subset V$ is given by $d(v,C)=\min
\{d(v,w);w\in C\}$ and that $[S]$ stands for the subgraph of $G$ induced by $%
S.$ To facilitate our discussion, a $t$-\textrm{PDDS} in $G$ all whose
components are isomorphic to a graph $H$ will be denoted by $t$-\textrm{PDDS}%
$[H].$\bigskip

\noindent Let $\Lambda _{n}$ be the integer grid, that is, \ $\Lambda _{n}$
has $Z^{n}$ as its vertex set with two vertices being adjacent if their
distance equals $1$. Let $P_{k}$ stand for a path on $k$ vertices. Thus $%
P_{1}$ is a single vertex and $P_{2}$ is a pair of adjacent vertices. It is
not difficult to see that an $r$-\textrm{PDDS}$[P_{1}]$ in $\Lambda _{n}$
constitutes a $PL(n,r)$ code while an $r$-\textrm{PDDS}$[P_{2}]$\textrm{\ }%
in $\Lambda _{n}$ is a $DPL(n,2r+2)$ code$.$ In \cite{ADH} it is conjectured
that:\bigskip

\begin{conjecture}
\label{G}Let $k\geq 1.$ Then there is no $r$-\textrm{PDDS}$[P_{k}]$ \ in $%
\Lambda _{n}$ for $n\geq 3$ and $r>1$ with the exception of $2$-\textrm{PDDS}%
$[P_{2}]$ in $\Lambda _{3}.$\bigskip
\end{conjecture}

\noindent As with the Golomb-Welch and the Etzion conjecture, if true, then
Conjecture \ref{G} is best possible. The existence of $t$-\textrm{PDDS}$%
[P_{k}]$ in $\Lambda _{n}$ for $n=2,k\geq 1,t\geq 1,$ and for $n\geq 3,k\geq
1,t=1$ has been proved in \cite{ADH}, while $2$-\textrm{PDDS}$[P_{2}]$ in $%
\Lambda _{3}$ is the $DPL(3,6)$ code mentioned above.\bigskip

\noindent At the end of this section we note that in \cite{ADH} there are
other results and a conjecture on the existence of $t$-\textrm{PDDS}$[H]$
sets in $\Lambda _{n}$ outside the scope of diameter perfect Lee codes.

\section{Lattice Tilings of $Z^{n}.$}

\noindent For the purpose of this paper it was convenient to define perfect
Lee codes by means of tilings. The main reason was that with this definition
in hand we will be able to prove the existence of certain codes using a
construction of lattice tilings through group homomorphism. We believe that
Stein in \cite{S1} was the first one to use graph homomorphisms to produce a
lattice tiling, in this case a tiling by crosses. Variations of Stein's
construction can be found in many papers, see \cite{ADH, S3, HS, HB, Molnar,
Sw,  S2,  S, Sz}, where the authors use Stein's approach for lattice tilings
by different types of crosses and spheres. In this section a generalization
of Stein's construction is given. We first describe this construction and
then we provide a necessary and sufficient condition that allows one to
verify whether there is a lattice tiling of $Z^{n}$ by a given
subset.\bigskip\ 

\noindent Let $\mathcal{L}$ be a lattice (=subgroup) of $Z^{n}$. Consider
the factor group $G=Z^{n}/\mathcal{L}$. Two elements $u,v\in Z^{n}$ belong
to the same coset of $Z^{n}/\mathcal{L}$ (that is, they are two different
representations of the same element $g\in G$) if $u-v\in \mathcal{L}$.
Choose one element $v_{g}\in Z^{n}$ from each coset of $Z^{n}/\mathcal{L}$%
\thinspace\ (=for each element of the factor group $G$) and put $%
V:=\{v_{g};v\in G\}.$ Then the collection $\mathcal{T=\{}V+l;$ $l\in 
\mathcal{L\}}$ forms a partition of $Z^{n}$ by translations of $V$. As $%
\mathcal{L}$ is a lattice $\mathcal{T}$ constitutes a lattice tiling of $%
Z^{n}$ by $V.$ A different choice of the elements $v_{g},$ although
considering the same lattice $\mathcal{L}$, provides tilings of $Z^{n}$ by
different sets $V.$ Consider a homomorphism $\phi :Z^{2}\rightarrow Z_{5},$
where $\phi (e_{1})=1,$ and $\phi (e_{2})=2$. In Fig.~1 it is shown that we
can choose $V$ to be a copy of the set $\{te_{i};0\leq t\leq 4\}$ as well as 
$V$ to be a copy of the Lee sphere $S_{2,1}.$ However, we are interested in
the "inverse" process. Given a set $V,$ find a lattice tiling $\mathcal{T=\{}%
V+l;l\in \mathcal{L\}}$ of $Z^{n}$ by $V.$ To be able to apply the above
construction we need to find a lattice $\mathcal{L}$ so that it would be
possible to choose elements $v_{g}$ from individual cosets of $Z^{n}/%
\mathcal{L}$ forming a copy of the set $V.$ However, to find a lattice with
the required properties might be a painstaking process. It is nearly
impossible to do it by trying all options as usually there are infinitely
many of them. To remedy the problem the following theorem provides a way how
to construct the required tiling $\mathcal{T}$ without knowing the lattice $%
\mathcal{L}$ explicitly. It turns out that all what is needed is to find a
suitable Abelian group $G$ of order $V$ and a homomorphism $\phi
:Z^{n}\rightarrow G$ satisfying a condition that can be easily verified. %
\begin{figure}[t]
\centering
\epsfig{figure=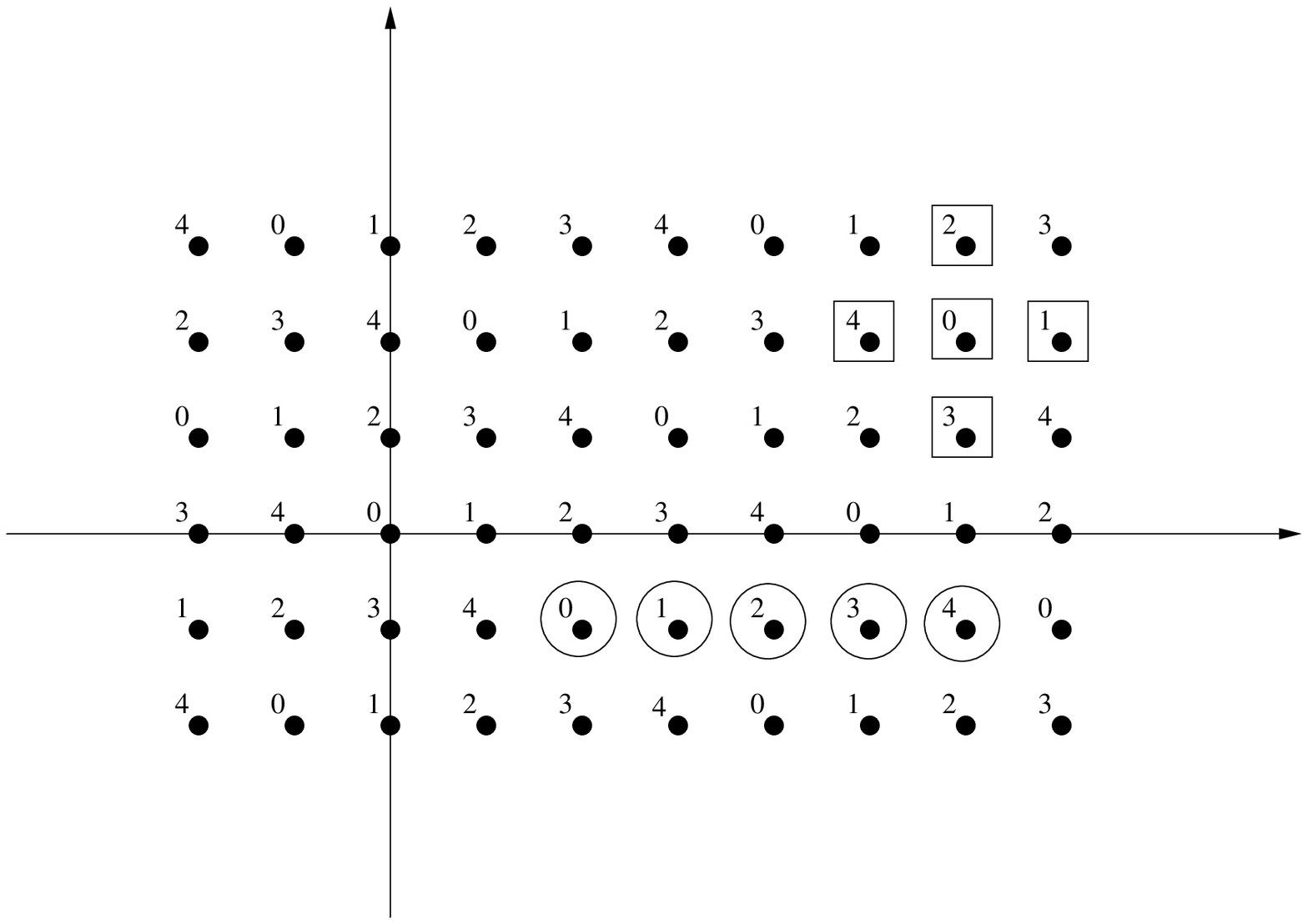, width=0.475\textwidth, angle=0}
\caption{Set $V$}
\label{fig 1}
\end{figure}

\begin{theorem}
\label{B}Let $V$ be a subset of $Z^{n}.$ Then there is a lattice tiling $%
\mathcal{T}$ of $Z^{n}$ by $V$ if and only if there is an Abelian group $G$
of order $\left\vert V\right\vert $ and a homomorphism $\phi
:Z^{n}\rightarrow G$ so that the restriction of $\phi $ to $V$ is a
bijection.\bigskip
\end{theorem}

\begin{proof}
Let $G$ be an Abelian group and let $\phi :Z^{n}\rightarrow G$ be a
homomorphism so that the restriction of $\phi $ to $V$ is a bijection.~It is
well known that $\ker (\phi )=\mathcal{L}$ is a subgroup of $Z^{n}$ and $%
Z^{n}/\mathcal{L}=G$. As mentioned above, two elements $x,y\in Z^{n}$ belong
to the same coset of the factor group $Z^{n}/\mathcal{L}$ iff $x-y\in 
\mathcal{L}.$ To show that $\mathcal{T=\{}V+l;l\in \mathcal{L}\}$ is a
tiling of $Z^{n}$ by $V$ we need to prove that (i) $\bigcup\limits_{l\in 
\mathcal{L}}V+l=Z^{n},$ (ii) $V+l\cap V+l^{\prime }=\emptyset $ for any $%
l,l^{\prime }\in \mathcal{L},l\neq l^{\prime }.$ Assume that there is $x\in
Z^{n}$ so that $x\notin \bigcup\limits_{l\in \mathcal{L}}V+l.$ Let $\phi
(x)=g.$ Since the restriction of $\phi $ to $V$ is a bijection, there is $%
y\in V$ so that $\phi (y)=\dot{g};$ that is, $x$ and $y$ belong to the same
coset of $Z^{n}/\mathcal{L}.$ Thus, $x-y=l\in \mathcal{L},$ which in turn
implies, as $y\in V,$ that $x\in V+l;l\in \mathcal{L}$, a contradiction. The
proof of (i) is complete. Assume now that there are $l,l^{\prime }\in 
\mathcal{L},l\neq l^{\prime },$ such that there is $x$ in $Z^{n}$ with $x\in
V+l\cap V+l^{\prime }$. Hence $x$ can be written as $x=y+l$ and also as $%
x=y^{\prime }+l^{\prime }$ for some $y\neq y^{\prime }\in V.$ We have $%
y+l=y^{\prime }+l^{\prime },$ that is $y-y^{\prime }=l^{\prime }-l\in 
\mathcal{L}.$ However, $y-y^{\prime }\in \mathcal{L}$ implies $\phi
(y-y^{\prime })=0,$ that is $\phi (y)=\phi (y^{\prime }),$ a contradiction
as $\phi $ is a bijection on $V$.\bigskip

\noindent To prove the necessary part of the condition, assume that $%
\mathcal{T=\{}V+l;l\in \mathcal{L}\}$ is a lattice tiling of $Z^{n}$ by $V;$
that is, $\mathcal{L}$ is a subgroup of $Z^{n}.$ Consider the factor group $%
G=Z^{n}/\mathcal{L}$ and the mapping $\phi :Z^{n}\rightarrow G$ given by $%
\phi (x)=g=[x]$, where $[x]$ is the coset of the factor group $Z^{n}/%
\mathcal{L}$ containing $x$. It is well known that $\phi $ is a
homomorphism. To finish the proof we need to show that the restriction of $%
\phi $ to $V$ is a bijection. Assume that there are $x,y\in V$ with $\phi
(x)=\phi (y).$ Then both $x,y$ belong to the same coset of $Z^{n}/\mathcal{L}%
,$ hence $x-y=l\in \mathcal{L},l\neq 0.$ This implies that $x\in V+l=V+(x-y)$
as $V$ contains $y.$ We arrived at a contradiction that $\mathcal{T}$ $%
=\{V+l;l\in \mathcal{L}\}$ is a tiling because $x\in V+O\cap V+l,$ where $%
l\neq O$. The proof is complete.\bigskip
\end{proof}

\noindent We recall that a set $\mathcal{S}\subset $ $Z_{\ast }^{n}$ is $p$%
-periodic, $p>0,$ if $s\in \mathcal{S}$ iff $s+pe_{i}\in \mathcal{S}$ for
all $i=1,...,n,$ and $p$ is the smallest number with the property.\bigskip

\begin{corollary}
\label{C}Let $\mathcal{T}$ \ be a lattice tiling of $Z^{n}$ by $V$ given by
the homomorphism $\phi :Z^{n}$\ $\rightarrow G.$ Then $\mathcal{T}$ $\ $is $%
p $-periodic, where $p=l.c.m.\{ord(\phi (e_{1})),...,ord(\phi (e_{n}))\}\,,$
where $ord(g)$ stands for the order of the element $g$ in the group $G$%
.\bigskip
\end{corollary}

\begin{proof}
Let $\mathcal{T=\{}V+l;l\in \mathcal{L\}}$, where $\mathcal{L=}\ker (\phi )$%
. It is well known that each lattice is periodic. To find the smallest
number $p$ for which $\mathcal{L}$ is periodic is equivalent to find the
smallest $p$ so that $pe_{i}\in \mathcal{L}$ for all $i=1,...,n.$ The
smallest number $p_{i}>0$ for which $p_{i}e_{i}\in \mathcal{L}$ equals the
smallest number for which $\phi (p_{i}e_{i})=\phi (e_{i})^{p_{i}}=0$ (in the
group $G$); i.e., equals $ord(\phi (e_{i})).$ Therefore, the smallest number 
$p$ so that $pe_{i}\in \mathcal{L}$ for all $i=1,...,n,$ equals $%
l.c.m.\{ord(\phi (e_{1})),...,ord(\phi (e_{n}))\}$.
\end{proof}

\noindent Now we list some properties of a homomorphism $\phi
:Z^{n}\rightarrow G$ in the above theorem that are useful to verify the
existence of a required tiling. The first of them claims that $\phi $ is
fully determined by images of $e_{i},i=1,...,n,$ under $\phi .$\bigskip

\begin{condition}
\label{1} If a mapping $\phi :Z^{n}\rightarrow (G,\circ )$ is a homomorphism
then, for any $a=(a_{1},...,a_{n})\in Z^{n},$ it is $\phi
((a_{1},...,a_{n})=\phi (e_{1})^{a_{1}}\circ ...\circ \phi (e_{n})^{a_{n}}$%
.\bigskip
\end{condition}

\begin{proof}
The proof follows from the fact that $\phi (u+v)=\phi (u)\circ \phi (v)$ for
any homomorphism $\phi $ and that $a=(a_{1},...,a_{n})$ can be written as $%
a=a_{1}e_{1}+...+a_{n}e_{n}.$\bigskip
\end{proof}

\noindent The following property is clearly true for all tilings and not
only for lattice ones. It allows to cut down significantly on the number of
cases if one wishes to prove the non-existence of a certain tiling.\ The
proof is obvious and therefore left to the reader.\bigskip

\begin{condition}
\label{2}If $\mathcal{T=\{}V+l;l\in \mathcal{L\}}$ is a tiling of $Z^{n}$ by 
$V$ then there is a tiling of $Z^{n}$ of the form $\{W+l;l\in \mathcal{L\}}$
for any translation $W$ of $V$.
\end{condition}

\noindent Based on Conditions \ref{1} and \ref{2} the following procedure
enables one to find a lattice tiling of $Z^{n}$ by a set $V$ or to prove the
non-existence of such tiling.

(i) Take arbitrary copy $W$ of $V.$ It is convenient to choose $W$ so that $%
O\in W,$ and $e_{i}\in W$ for as many $i=1,...,n$ as possible.

(ii) Choose an Abelian group $G$ of order $\left\vert V\right\vert .$

(iii) Choose a homomorphism $\phi :Z^{n}\rightarrow G$; i.e., choose $n$
elements of $G$ as $\phi (e_{i}).$ Check whether the restriction of $\phi $
to $W$ is a bijection.

\noindent As there are finitely many Abelian groups of a fixed finite order
the above procedure is finite. The computational complexity depends on the
size of $V,$ and on the number of Abelian groups of the given order.\bigskip

\noindent We illustrate the procedure by two simple examples. In Fig.~1,
it is shown how to choose $\phi (e_{i})$ to get a tiling of $Z^{2}$ into Lee
spheres $S_{2,1}.$ Consider now the set $V$ $%
=\{O,e_{1},2e_{1},e_{2},e_{1}-e_{2}\}.$ We will show that there is no
lattice tiling of $Z^{n}$ by $\dot{V}.$ Suppose that there is a lattice
tiling by $V$. Then there is a homomorphism $\phi :Z^{n}\rightarrow G$ with
required properties. Clearly $G=Z_{5}$ as there is only one Abelian group of
order $5.$ To cut down on the number of cases we note that wlog we can set $%
\phi (e_{1})=1.$ Indeed, all elements of $Z_{5}$ are its generators,
therefore there exists an automorphism $\psi $ on $Z_{5}$ with $\psi (1)=a$
for any $a\neq 0\in Z_{5}.$ We are left with two choices, namely $3$ and $4,$
for $\phi (e_{2}).$ However, in both case the restriction of $\phi $ to $V$
is not a bijection, see Fig.~2. Thus, there is no lattice tiling of $Z^{n}$
by $V.$\bigskip %
\begin{figure}[t]
\centering
\epsfig{figure=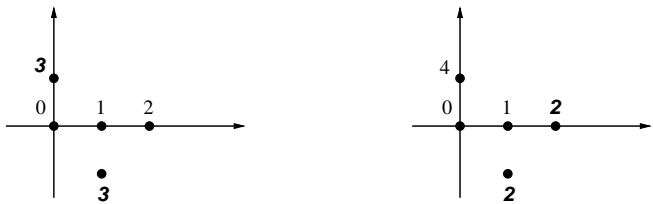, width=0.475\textwidth, angle=0}
\caption{The restriction of $\protect\phi $ to $V$}
\label{fig 2}
\end{figure}

\noindent To be able to enumerate all - not only lattice - diameter-$4$
perfect Lee codes in $Z^{n}$, we will describe a construction of non-lattice
tilings of $Z^{n}$. It turns out that a variation of Theorem \ref{B} might
be even used as a starting point to produce such tilings of $Z^{n}$.$\,\ $To
do so we need to consider a generalization of the problem to a super-lattice
of $Z^{n}$. This is a quite common situation in mathematics when a problem
has to be generalized before we are able to find its solution. To have a
geometric interpretation of this generalization, instead of tiling $Z^{n}$
by a finite set of words $V$, we will consider tilings of Euclidean $n$%
-space $\mathcal{%
\mathbb{R}
}^{n}$ by clusters of unit cubes. In this setting a tiling of $Z^{n}$ by $V$
generates a tiling of $\mathcal{%
\mathbb{R}
}^{n}$ by a cluster $C$ of unit cubes so that each word of $V$ is a center
of a unit cube in $C$. E.g., a tiling of $Z^{n}$ by a sphere $S_{n,r}$
generates a tiling of $\mathcal{%
\mathbb{R}
}^{n}$ by cubistic cross-polytopes, see Fig.~3 for $n=2,3,$ and $r=1,2$. We
note that the cubistic cross polytope with $r=1$ is also called $n$-cross.
Clearly, a tiling $\mathcal{T}$ of $\mathcal{%
\mathbb{R}
}^{n}$ by $C$ generated by a tiling of $Z^{n}$ has the following property:
Let $T,T^{\prime }$ be two tiles in $\mathcal{T}$ with cubes $C\in T$ and $%
C^{\prime }\in T^{\prime }$. If $C$ and $C^{\prime }$ have a non-empty $%
(n-1) $-dimensional intersection then $C$ and $C^{\prime }$ are neighbors;
this type of tiling is also called regular or a face-to-face tiling. We
recall that two cubes $C,C^{\prime }$ are neighbors if their centers $%
c=(c_{1},...,c_{n}),c^{\prime }=(c_{1}^{\prime },...,c_{n}^{\prime })$ have
the following property: there is an index $i$ with $\left\vert
c_{i}-c_{i}^{\prime }\right\vert =1$ and $c_{j}=c_{j}^{\prime }$\thinspace\
for $i\neq j$. Thus, a non-regular tiling $T=\{T_{i},i\in I\}$ of $\mathcal{%
\mathbb{R}
}^{n}$ by \ a cluster of cubes $C$ contains two cubes $C\in T$ and $%
C^{\prime }\in T^{\prime },T\neq T^{\prime }$ so that $C$ and $C^{\prime }$
have a non-empty $(n-1)$-dimensional intersection but they are not
neighbors. See Fig.~4 for an example of two unit cubes with a non-empty $%
(n-1)$ dimensional intersection that are not neighbors for $n=2,3.$ It turns
out that non-regular lattice tilings of $\mathcal{%
\mathbb{R}
}^{n}$ by a suitable cluster of cubes will allow us to construct
non-lattice, and even non-periodic tilings of $Z^{n}$ by spheres $S_{n,1}$
and double-spheres $DS_{n,1}.$ We will use the following straightforward
generalization of Theorem \ref{B}.

%
\begin{figure}[!t]
\centering
\subfigure[$n = 2$, $r = 1$] {
\epsfig{figure=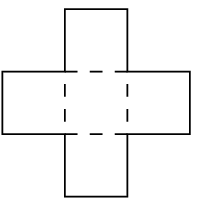, width=0.14\textwidth, angle=0}
} 
\subfigure[$n = 2$, $r = 2$] {
\epsfig{figure=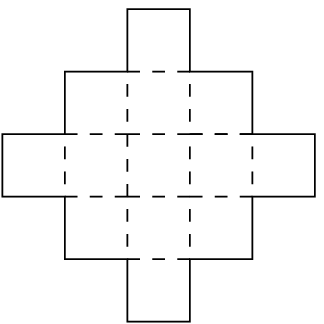, width=0.18\textwidth, angle=0}
} 
\subfigure[$n = 3$, $r = 1$] {
\epsfig{figure=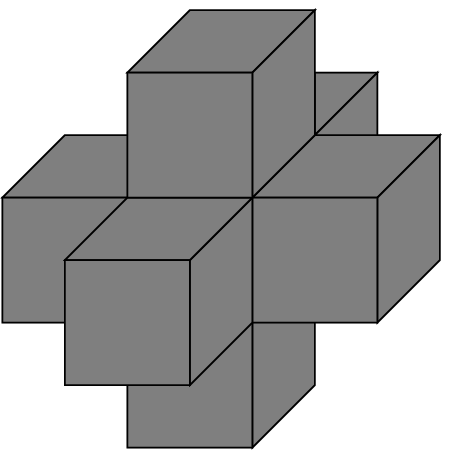, width=0.15\textwidth, angle=0}
} 
\subfigure[$n = 3$, $r = 2$] {
\epsfig{figure=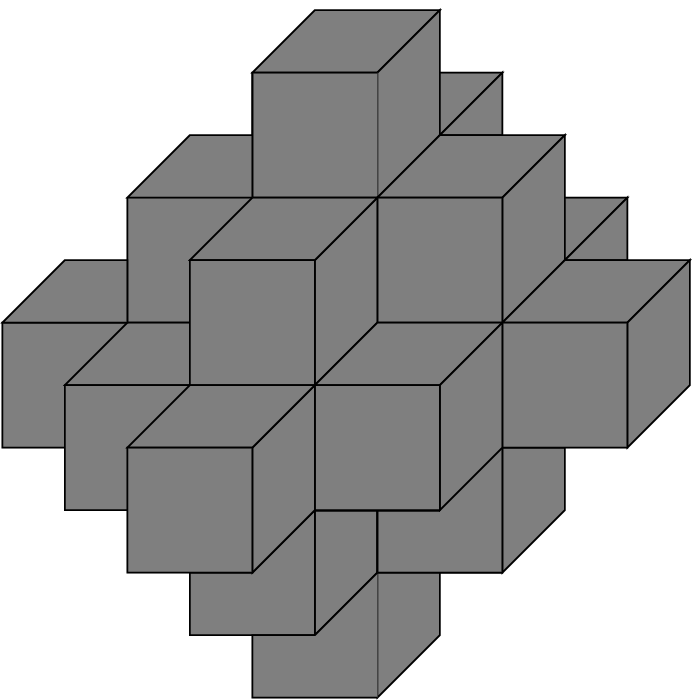, width=0.2\textwidth, angle=0}
}
\caption{Cubistic cross-polytopes}
\label{fig 3}
\end{figure}
%
%
\begin{figure}[!t]
\centering
\subfigure[$n = 2$] {
\epsfig{figure=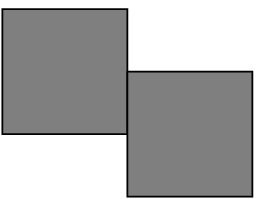, width=0.14\textwidth, angle=0}
} 
\subfigure[$n = 3$] {
\epsfig{figure=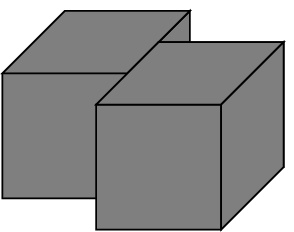, width=0.16\textwidth, angle=0}
}
\caption{Cubes having non-empty$(n-1)$ dimensional intersection that are not
neighbors}
\label{fig 4}
\end{figure}

\bigskip

\begin{theorem}
\label{D}Let $V$ be a subset of a lattice $\mathcal{M\subset 
\mathbb{R}
}^{n}.$ Then there exists a lattice tiling $\mathcal{T}$ \ of $\mathcal{M}$
by $V$ if and only if there is an Abelian group $G$ of order $\left\vert
V\right\vert $ and a homomorphism $\phi :\mathcal{M}\rightarrow G$ so that
the restriction of $\phi $ to $V$ is a bijection.\bigskip
\end{theorem}

\noindent Theorem \ref{D} can be proved by mimicking the proof of Theorem %
\ref{B}. To construct a non-regular tiling we will choose $\mathcal{M}$ to
be the lattice generated by vectors $\frac{1}{2}e_{1},e_{2},...,e_{n}.$ Then
a tiling $\mathcal{T}$ of $\mathcal{%
\mathbb{R}
}^{n}$ by a cluster of cubes with their centers in $\mathcal{M}$ will be
non-regular if and only if there are in $\mathcal{T}$ a cube $C$ with its
center at an integral point and a cube $C^{\prime }$ with its center at a
point whose first coordinate is of form $m+0.5,m\in Z$. Theorem \ref{E}
provides a sufficient condition for the existence of such non-regular tiling
of $\mathcal{%
\mathbb{R}
}^{n}$ by a cluster of cubes. \bigskip

\begin{theorem}
\label{E}Let $C$ be a cluster of cubes centered at $V\subset Z^{n},$ and $%
\mathcal{M}$ be a lattice generated by vectors $\frac{1}{2}%
e_{1},e_{2},...,e_{n}.$ Then there is a lattice non-regular tiling of $%
\mathcal{%
\mathbb{R}
}^{n}$ by $C$ if there is a homomorphism $\phi :\mathcal{M\rightarrow }G,$
an Abelian group of order $2\left\vert V\right\vert ,$ so that the
restriction of $\phi $ to $W=V\cup \{V+\frac{1}{2}e_{1}\}$ is a bijection
and $\phi (e_{i})$ is a generator of $G$ for some $i\geq 2.$\bigskip\ \ \ \ 
\end{theorem}

\begin{proof}
The homomorphism $\phi $ determines a tiling $\mathcal{T=\{}W+l;l\in 
\mathcal{L\}}$ of $\mathcal{M}$ by $W,$ where $\mathcal{L}=\ker (\phi )$ is
a lattice. \ Clearly, $\mathcal{T}$ generates a tiling of $\mathcal{%
\mathbb{R}
}^{n}$ by the cluster $C$ of cubes centered in $V$. Indeed, split each cube
in $C$ by a hyperplane orthogonal to $e_{1}.$ Then we get a collection of $%
2\left\vert C\right\vert $ half-cubes centered at points of $W=V\cup \{V+%
\frac{1}{2}e_{1}\}$ in $\mathcal{M}.$ Thus both a tiling $\mathcal{F}$ of $%
Z^{n}$ by $V$ and a tiling $\mathcal{F}^{\prime }\,\ $of $\mathcal{M}$ by $%
W=V\cup \{V+\frac{1}{2}e_{1}\}$ they generate a tiling of $\mathcal{%
\mathbb{R}
}^{n}$ by the cluster $C$. The only difference between $\mathcal{F}$ and $%
\mathcal{F}^{\prime }$ is that while the centers of cubes of $\mathcal{F}$
are in $Z^{n},$ the centers of the half-cubes in $\mathcal{F}^{\prime }$ are
in $\mathcal{M}.$ To prove that $\mathcal{T}$ generates a non-regular tiling
of $\mathcal{%
\mathbb{R}
}^{n}$ it suffices to show that there are $l,l^{\prime }$ in $\mathcal{L}$
so that $l_{1}$ is an integer while $l_{1}^{\prime }$ is of the form $%
m+0.5,m\in Z\,,$ where $l_{1}$ and $l_{1}^{\prime }$ are the first
coordinates of $l$ and $l^{\prime },$ respectively. Indeed, then the centers
of cubes in $C+l$ will be at integer points, while the centers of cubes in $%
C+l^{\prime }$ will not. Since $\phi (e_{i}),$ for some $i\geq 2,$ is a
generator of $G$ then there are numbers $a$,$b\in Z$ so that $\phi
(ae_{i})=0 $ and $\phi (\frac{1}{2}e_{1}+be_{i})=0$. Thus, the two vectors $%
l=ae_{i},l^{\prime }=\frac{1}{2}e_{1}+be_{i}\in \mathcal{L}$ have the
required property.
\end{proof}

\section{Diameter perfect Lee codes}

\noindent This section contains main results of our paper. We focus on the
existence and enumeration of diameter perfect Lee codes over $Z_{q}^{n}$ and 
$Z^{n}.$ It is conjectured in \cite{E} and \cite{ADH}, see Section 2, that
for $n\geq 3$ these codes exist only when $d=3,4,$ with the exception of $%
DPL(3,6)$ code constructed by Minkowski. We concentrate on linear codes as
they are most important from the practical point of view. All values of $q$
for which there exists a linear $PL(n,1,q)$ code (=linear $DPL(n,3,q)$ code)
have been found in \cite{H}. Now we determine all values of $q\,$for which
there is a linear $DPL(n,4,q)$ code. Clearly, it suffices to determine all
values $q$ for which there exists a linear \textit{non-periodic} $DPL(n,d,q)$
code as all other linear $DPL(n,d,q)$ codes can be obtained by a periodic
extension of a non-periodic one. \bigskip

\noindent The following theorem claims that instead of investigating
diameter perfect codes of $Z_{q}^{n}$ one can confine himself to codes in $%
Z^{n}$. The statement is intuitively clear and in some papers the authors
take validity of analogous cases for granted. We will sketch its proof. We
recall that a set $\mathcal{S}\subset $ $Z_{\ast }^{n}$ is $p$-periodic, $%
p>0,$ if $s\in \mathcal{S}$ iff $s+pe_{i}\in \mathcal{S}$ for all $%
i=1,...,n, $ and $p$ is the smallest positive number with the
property.\bigskip

\begin{theorem}
\label{A}For $d$ is even, there is a linear $DPL(n,d,q)$ code if and only if
there is a linear $DPL(n,d)$ $p$-periodic code with $p|q.$\bigskip
\end{theorem}

\begin{proof}
We start with the following auxiliary statement.\bigskip

\noindent \textit{Lemma. }Let $\mathcal{L}$ be a linear diameter-$d$ perfect
Lee code in $Z_{\ast }^{n}$ and $W$ be a double-sphere $DS_{n,r}$ with $r=%
\frac{d-2}{2}.$ Then $\mathcal{T}=\{W+l;l\in \mathcal{L}\}$ is a lattice
tiling of $Z_{\ast }^{n}$ by double-spheres.\bigskip

\noindent First we prove this lemma. Assume that for $l\neq l^{\prime },$
there is $x\in Z_{\ast }^{n}$ such that $x\in W+l\cap W+l^{\prime }.$ Then
there are $w,w^{\prime }\in W$ so that $x=w+l=w^{\prime }+l^{\prime }.$ This
in turn implies $\left\vert w-w^{\prime }\right\vert =\left\vert l^{\prime
}-l\right\vert $ which is a contradiction since $\left\vert l-l^{\prime
}\right\vert =\rho _{L}(l,l^{\prime })\geq d$ for any $l,l^{\prime }\in 
\mathcal{L}$ as $\mathcal{L}$ is a diameter-$d$ code while $\left\vert
w-w^{\prime }\right\vert =\rho _{L}(w,w^{\prime })<d$ for any $w,w^{\prime
}\in W$ because $W$ is an anticode of diameter $d-1.$ To finish this part of
the proof we need to show that $\mathop{\textstyle \bigcup }\limits_{l\in 
\mathcal{L}}W+l=Z_{\ast }^{n}.$ This is obvious for $Z_{\ast }^{n}=Z_{q}^{n}$
as the space is finite and $\mathcal{L}$ is a transversal of a tiling by
double-spheres. If $\mathcal{L}$ is a linear code then the tiling is a
lattice tiling. Assume now that $\mathcal{L}$ is a linear $DPL(n,d)$ code,
and $\mathop{\textstyle \bigcup }\limits_{l\in \mathcal{L}}W+l\subsetneqq
Z^{n}.$ Then the volume $V(\mathcal{L)}$ of the lattice $\mathcal{L}$ would
be strictly bigger than the volume of the double-sphere $DS_{n,r}$ which is
a contradiction as $\mathcal{L}$ is a transversal of a tiling of $Z^{n}$ by
double-spheres. The proof of Lemma is complete.\bigskip

\noindent Let $\mathcal{L}$ be a linear $DPL(n,d,q)$ code. Then $\rho
_{L}(l,l^{\prime })\geq d$ for all $l,l^{\prime }\in \mathcal{L},$ and by
the above lemma $T=\{W+l;l\in \mathcal{L\}}$ is a lattice tiling of $%
Z_{q}^{n}$ by double-spheres $DS_{n,r},r=\frac{d-2}{2}.$ Consider an
extension of $\mathcal{L}$ given by $\mathcal{L}^{\prime }=\{l^{\prime
};l^{\prime }=l+\sum\limits_{i=1}^{n}k_{i}qe_{i},$ where $l\in \mathcal{L},$
and $k_{i}\mathcal{\in }Z$ for all $i=1,..,n\}.$ It is a matter of technical
routine to prove that $\mathcal{L}^{\prime }$ is a lattice as well with $%
\rho _{L}(l,l^{\prime })\geq d$ for all $l,l^{\prime }\in \mathcal{L}%
^{\prime }$, of period $q,$ and $\mathcal{T}^{\prime }=\{W+l^{\prime
};l^{\prime }\in \mathcal{L}^{\prime }\}$ is a lattice tiling of $Z^{n}$.
Therefore, $\mathcal{L}^{\prime }$ is a linear $DPL(n,d)$ code that is $p$%
-periodic for some $p|q.$ On the other hand, let $\mathcal{L}$ be a linear $%
DPL(n,d)$ $q$-periodic code. Then, $\rho _{L}(l,l^{\prime })\geq d$ for all $%
l,l^{\prime }\in \mathcal{L},$ and by our lemma, $T=\{W+l;l\in \mathcal{L\}}$
is a lattice tiling of $Z^{n}.$ Let $\mathcal{L}^{\prime }=\mathcal{L}\cap
Z_{q}^{n}$ be a restriction of $\mathcal{L}$ to $Z_{q}^{n}.$ Clearly, $\rho
_{L}(l,l^{\prime })\geq d$ for every $l,l^{\prime }\in \mathcal{L}^{\prime
}. $ Further, as above, it is a matter of routine to verify that $\mathcal{L}%
^{\prime }$ is a lattice and $\{W+l;l\in \mathcal{L}^{\prime }\}$ is a
tiling of $Z_{q}^{n}$ by double-spheres $DS_{n,r},r=\frac{d-2}{2}.$ Thus, $%
\mathcal{L}^{\prime }$ is a linear $DPL(n,d,q)$ code.\bigskip
\end{proof}

\noindent The previous theorem claims that instead of showing the existence
of a distance-$d$ perfect linear Lee code over $Z_{q}^{n}$ (=a tiling of $%
Z_{q}^{n}$ by double-spheres) it suffices to prove the existence of a
linear, $p$-periodic, $p|q,$ distance-$d$ perfect Lee code over $Z^{n}$ (=a
tiling of $Z^{n}$ by double-spheres)$.$ The difference in difficulty of
proving the existence of a lattice tiling of $Z_{q}^{n}$ and of $Z^{n}$
might seem to be negligible. However, when looking for a tiling of $Z^{n}$
one has at its disposal a powerful construction based on groups homomorphism
described in the previous section.\bigskip

\noindent We are ready to determine all $q$ for which there exists a
non-periodic linear $DPL(n,4,q)$ code.\bigskip

\begin{theorem}
Let $n=2^{\alpha }p_{1}^{\alpha _{1}}...p_{k}^{\alpha _{k}}$ be the prime
number factorization of $n,$ where possibly $\alpha =0,$ and $p_{i}>2$ for $%
i=1,...,k.$ Set $p=p_{1}...p_{k}.$ Then a linear non-periodic $DPL(n,4,q)$
code exists if and only if $q=2^{\beta }p_{1}^{\beta _{1}}...p_{k}^{\beta
_{k}},$ where $2\leq \beta \leq \alpha +2,$ and $1\leq \beta _{i}\leq \alpha
_{i}$. In particular, the smallest $q$ for which there exists a linear
non-periodic $DPL(n,4,q)$ code equals $4p.$\bigskip
\end{theorem}

\begin{proof}
Assume that there is a linear non-periodic $DPL(n,4,q)$ code. By Theorem \ref%
{A} there is a linear $q$-periodic $DPL(n,d)$ code $\mathcal{L},$ and by the
lemma stated in the proof of Theorem \ref{A}, $\mathcal{T=}\{V+l;l\in 
\mathcal{L}\}$ is a lattice tiling of $Z^{n}$ by double-spheres $DS_{n,1}.$
As a copy $V$ of $DS_{n,1}$ we choose $V=\{\pm e_{i},\pm
e_{i}+e_{1};i=1,...,n\}.$ Clearly, $\left\vert V\right\vert =4n.$ By Theorem %
\ref{B} there is a homomorphism $\phi :Z^{n}\rightarrow G,$ an Abelian group
of order $4n,$ such that the restriction of $\phi $ to $V$ is a bijection.
First we show that $4p|q.$ We note that $4|ord(\phi (e_{1}))$. Indeed,
consider the subgroup $H$ generated by $\phi (e_{1}).$ Then it is easy to
check that for each $i,1\leq i\leq n$, $\left\vert \phi (\pm e_{i},\pm
e_{i}+e_{1})\cap H\right\vert =0$ or $4.$ Note that this is true also for $%
i=1.$ Thus $4|ord(\phi (e_{1})),$ which in turn implies $4|q,$ see Corollary %
\ref{C}$.$ Now we prove that $p|q.$ Let $G=Z_{t_{1}}\times Z_{t_{2}}\times
...\times Z_{t_{s}},$ be the factorization of $G,$ where possibly $s=1;$
that is, possibly $G$ is a cyclic group. We recall that $n=2^{\alpha
}p_{1}^{\alpha _{1}}...p_{k}^{\alpha _{k}},$ so $4n=2^{\alpha
+2}p_{1}^{\alpha _{1}}...p_{k}^{\alpha _{k}}.$ We show that, for every $%
j=1,...,k,$ there is $i,1\leq i\leq n,$ depending only on $j$ so that $%
p_{j}|ord(\phi (e_{i})).$ By Corollary \ref{C}, this will in turn imply that 
$p|q.$ As $4n=t_{1}\times ...\times t_{s},$ there is $t_{i}$ so that $%
p_{j}|t_{i};$ wlog assume $p_{j}|t_{1}.$ Let $g=\phi
(e_{1})=(g_{1},g_{2}...,g_{s}).$ We consider two cases. If $p_{j}\nmid
g_{1}, $ then $p_{j}|ord(g)$ and we are done. Otherwise, for $p_{j}|g_{1},$
consider the element $f=(1,0,...,0).$ Since the restriction of $\phi $ to $V$
is a bijection, there is $i$ so that $f\in \phi (\{\pm e_{i},\pm
e_{i}+e_{1}\}.$ If $f=\phi (e_{i})$ or $f=\phi (e_{i})^{-1}$ we are done as $%
p_{j}|ord(\phi (e_{i}))=t_{1}$. Finally we are left with two cases, either $%
f=\phi (e_{i}+e_{1})=\phi (e_{1})+g$ or $f=\phi (-e_{i}+e_{1})=-\phi
(e_{i})+g$ for some $i$. In the former case, $\phi
(e_{i})=f-g=(1-g_{1},-g_{2},...,-g_{s}),$ in the latter case $\phi
(e_{i})=g-f=(-1+g_{1},g_{2},...,g_{s}).$ As $p_{j}|g_{1}$ and $p_{j}|t_{1}$
we have $p_{j}\nmid \pm (1-g_{1}),$ that is $p_{j}|ord(\phi (e_{i}))$ in
both cases. Thus we proved that $p_{j}|q,$ which implies that $p|q.$ In
aggregate, $4p|q.$ It is obvious that $q\leq 4n.$ Indeed, by Corollary \ref%
{C}, $\mathcal{T}$ is $q$-periodic where $q=l.c.m.(ord(\phi
(e_{1})),...,ord(\phi (e_{n}))).$ As the order of each element divides the
order of the group also $q$ does, so $q\leq ord(G)=4n,$ which in aggregate
implies that $q=2^{\beta }p_{1}^{\beta _{1}}...p_{k}^{\beta _{k}},$ where $%
2\leq \beta \leq \alpha +2,$ and $1\leq \beta _{i}\leq \alpha _{i}$.\bigskip

\noindent Now we show that there is a linear non-periodic $DPL(n,4,q)$ code
for each $q=2^{\beta }p_{1}^{\beta _{1}}...p_{k}^{\beta _{k}},$ where $2\leq
\beta \leq \alpha +2,$ and $1\leq \beta _{i}\leq \alpha _{i}$. \ By Theorems %
\ref{X} and \ref{A}, it suffices to construct a $q$-periodic, lattice tiling 
$\mathcal{T}$ $=\{V+l;l\in \mathcal{L}\}$ of $Z^{n}$ by double-spheres so
that, for any $l\in \mathcal{L},$ the Lee weight $\left\vert l\right\vert $
of $l$ is even. This tiling $\mathcal{T}$ will be constructed by means of
Theorem \ref{B}. Hence it suffices to find a homomorphism $\phi
:Z^{n}\rightarrow G$, where $G$ is an Abelian group of order $4n,$ so that
the restriction of $\phi $ to a $V$ is a bijection. To guarantee that $%
\mathcal{T}$ is $q$-periodic, we choose $G$ so that $ord(g)|q$ for each $%
g\in G$, see Corollary \ref{C}, and we choose $\phi (e_{1})$ so that $%
ord(\phi (e_{1}))=q.$\bigskip

\noindent To have a group $G$ of order $4n$ with the property $ord(g)|q$ for
each $g\in G$ we set $G=Z_{t_{1}}\times Z_{t_{2}}\times ...\times Z_{t_{s}},$
where $t_{1}=q,$ and $t_{i}=2^{\gamma }p_{1}^{\gamma _{1}}...p_{k}^{\gamma
_{k}},$ where $\gamma ,\gamma _{j}\in \{0,1\}$ for all \thinspace $i>1,1\leq
j\leq k.$\bigskip

\noindent The homomorphism $\phi $ is defined by choosing $g_{i}:=\phi
(e_{i}).$ First we set $g_{1}=(1,0,...,0)$ and $g_{i}=(2i-1,0,...,0)$ for $%
2\leq i\leq \frac{q}{4}.$ It is easy to check that $\phi (\{\pm e_{i},\pm
e_{i}+e_{1};i=1,...,\frac{t}{4}\}=Z_{q}\times \{0\}\times ...\times \{0\}.$%
\bigskip

\noindent Similarly, for $b=(0,b_{2},...,b_{s})\in G$, $b\neq 0,$ with $%
ord(b)=2,$ we set $g_{i+j}=(2i-1,b_{2},...,b_{s}),i=1,...,\frac{q}{4},$ for
a suitable fixed $\dot{j}.$ Clearly, $\phi (\{\pm e_{i+j},\pm
e_{i+j}+e_{1};i=1,...,\frac{q}{4}\}=$ $Zq\times \{b_{2}\}\times ...\times
\{b_{n}\}.$\bigskip

\noindent Finally, if $b=(0,b_{2},...,b_{s})\in G$, $b\neq 0,$ with $%
ord(b)\neq 2,$ we associate $\frac{q}{2}$ dimensions with the elements $b$
and $b^{-1\,\ }$ by $g_{i+j}=(2i-1,b_{2},...,b_{s}),i=1,...,\frac{q}{2},$
where $\dot{j}$ is a suitable fixed number. Then $\phi (\{\pm e_{i+j},\pm
e_{i+j}+e_{1};i=1,...,\frac{q}{2}\}=$ $Z_{q}\times \{b_{2}\}\times ...\times
\{b_{n}\}\cup Z_{q}\times \{-b_{2}\}\times ...\times \{-b_{n}\}.$ It is easy
to see that in aggregate $\phi (V)=G.$ So the restriction of the mapping $%
\phi $ to $V$ is a bijection. Thus we proved that there exists a lattice
tiling $\mathcal{T}$ of $Z^{n}$ by double-spheres that is $q$-periodic.
\bigskip

\noindent Denote the lattice in the tiling $\mathcal{T}$ by $\mathcal{L}.$
We will prove that, for each $l\in \mathcal{L},$ the Lee weight $\left\vert
l\right\vert $ is even, which in turn implies, by Theorem \ref{X}, that $%
\mathcal{L}$ is a linear $DPL(n,d)$ code. To prove it we will construct a
vector basis $v_{i},i=1,...,n,$ of $\mathcal{L}$ whose each element is of
even weight. First we set $v_{1}=qe_{1}.$ As $q$ is even $\left\vert
v_{1}\right\vert =q$ is even, and $\phi (v_{1})=(q,0,...,0)=(0,...,0);$ i.e. 
$v_{1}\in \ker (\phi )=\mathcal{L}.$ Further, we set $%
v_{i}=(2i-1)e_{1}-e_{i} $ for $i=2,...,\frac{q}{4}.$ Also in this case we
get $\phi (v_{i})=(2i-1,0,..,0)-(2i-1,0,...,0)=(0,...,0)$ and $\left\vert
v_{i}\right\vert =2i-2.$ Now, for each $j=2,...,s$ there is $m_{j}$ so that $%
\phi (e_{m_{j}})=(1,0,...0,1,0,...0)$ with $j$-th coordinate equal to $1.$
We set $v_{m_{j}}=t_{j}e_{1}-t_{j}e_{m_{j}}.$ Again, as before, $\phi
(v_{m_{j}})=(0,...,0)$ and $\left\vert v_{m_{j}}\right\vert =2t_{j}$ is
even. Finally, if $\phi (e_{i})=(b_{1},b_{2},...b_{s}),$ where $i>\frac{q}{4}%
,i\neq m_{j},j=2,...,s,$ then we set $%
v_{i}=(b_{1}-b_{2}-...-b_{s})e_{1}+b_{2}e_{m_{2}}+...+b_{s}e_{m_{s}}-e_{i}.$
We get $\phi (v_{i})=(b_{1}-b_{2}-...-b_{s},0,...,0)+(b_{2},b_{2},0,...,0)+$ 
$(b_{3},0,b_{3},0,...,0)+...+$ $%
(b_{s},0,...,0,b_{s})-(b_{1},b_{2},...b_{s})=(0,...,0)$ and $\left\vert
v_{i}\right\vert =b_{1}+2(b_{2}+...+b_{s})+1$ which is even as $b_{1}$ is
always odd, see definition of $g_{i}=\phi (e_{i}).$ Thus we showed that, for
all $i=1,...,n,$ $v_{i}$ belongs to $\ker \phi =\mathcal{L}$ and $\left\vert
v_{i}\right\vert $ is even. It is not difficult to see that $%
v_{i},i=1,...,n, $ are independent vectors. Let $A$ be a matrix with $v_{i}$
being its $i$-th row. After permuting columns of $A$ so that the columns $%
m_{j},j=2,...,s,$ become columns indexed $\frac{q}{4}+1,...,\frac{q}{4}+s,$
we get a lower triangular matrix with its diagonal equal to $%
t_{1},-1,...,-1,-t_{2},-t_{3},...,-t_{s},-1,...,-1.$ Therefore $\left\vert
\det A\right\vert =t_{1}t_{2}...t_{s}=4n,~$which in turn implies that $v_{i}$%
's form a basis with the required properties. The proof is complete.\bigskip
\end{proof}

\noindent A cluster consisting of $2n+1$ unit cubes centered at points of
Lee sphere $S_{n,1}$ is called the $n$-cross. \ By the $n$-double-cross we
understand a cluster of $4n$ unit cubes centered at points of double-sphere $%
DS_{n,1},$ see Fig.~5 for $3$-double-cross. Szabo \cite{Sz} proved that there
is a non-regular lattice tiling of $%
\mathbb{R}
^{n}$ by $n$-crosses if and only if $2n+1$ is not a prime. The non-existence
of such lattice tiling for $2n+1$ being a prime follows from Redei's
decomposition theorem see \cite{Redei}. We show that, for some $n,$ there is
a non-regular lattice tiling of $%
\mathbb{R}
^{n}$ by $n$-double-crosses. At the first glance it seems that non-regular
tilings are not useful in coding theory. However, we will show that they are
instrumental for enumerating $DPL(n,4)$ codes.

%
\begin{figure}[!t]
\centering
\epsfig{figure=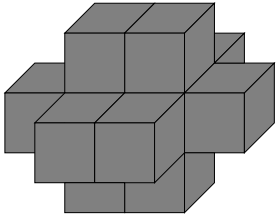, width=0.2\textwidth, angle=0}
\caption{$3$-double-cross}
\label{fig 5}
\end{figure}

\bigskip

\begin{theorem}
\label{F}Let $n>1$ be a natural number that is not a power of $2.$ Then
there exists a non-regular lattice tiling of $%
\mathbb{R}
^{n}$ by $n$-double-crosses.\bigskip
\end{theorem}

\begin{proof}
As above we use the double sphere $V=\{\pm e_{i},\pm
e_{i}+e_{1};i=1,...,n\}. $ By Theorem \ref{E} it suffices to find an Abelian
group $G$ of order $8n$ and a homomorphism $\phi :\mathcal{M}\rightarrow G$
so that the restriction of $\phi $ to $W=V\cup \{V+\frac{1}{2}e_{1}\}$ is a
bijection and $\phi (e_{i})$ is a generator of $G$ for some $i>1.$ Since $\
n $ is not a power of $2,$ then $n$ can be written in the form $%
n=2^{t}(2k+1),k>0.$ As $G$ we choose the cyclic group $Z_{8n}.$ We set

$\phi (\frac{1}{2}e_{1})=2k+1,$ and, for $t>0,$ $\phi
(e_{i})=(4i-2)(2k+1),i=2,...,2^{t}.$ This way $\phi
(\bigcup\limits_{i-1}^{2^{t}}\{\pm e_{i},\pm e_{i}+e_{1}\}\cup \{\{\pm
e_{i},\pm e_{i}+e_{1}\}+\{\frac{1}{2}e_{1}\}\})$ $=\{s;0\leq s<8n,$ $s$ $%
\equiv 0(mod)(2k+1)\}.$\bigskip

\noindent For $j=1,...,k,$ we set $\phi (e_{i+c_{j}})=j+4(i-1)(2k+1)$ for $%
i=1,...,2^{t+1},$ where $c_{j}=2^{t}+(j-1)2^{t+1}.$ For fixed $j$ we have $%
\phi (\bigcup \{\pm e_{i},\pm e_{i}+e_{1}\}\cup \{\{\pm e_{i},\pm
e_{i}+e_{1}\}+\{\frac{1}{2}e_{1}\}\}),$ where $2^{t}+(j-1)2^{t+1}+1\leq
i\leq 2^{t}+j2^{t+1})$ $\ =\{s;0\leq s<8n,$ $s\equiv \pm j(mod)$ $(2k+1)\}.$
Thus, in aggregate, $\phi (V)=G.$ To finish the proof it is sufficient to
notice that $\phi (e_{1+2^{t}})=1$, a generator of $G.$ The proof is
complete.\bigskip
\end{proof}

\noindent As mentioned in the introduction the double-sphere $DS_{3.1}$ can
be seen as the union of two Lee spheres $S_{3,1}$ with centers at distance $%
1.$ Therefore it is somewhat surprising that there exists a non-regular
tiling $\mathcal{T}$ of $%
\mathbb{R}
^{3}$ by $3$-double-cross although there is no non-regular tiling of $%
\mathbb{R}
^{3}$ by $3$-crosses. A detailed description of the tiling $\mathcal{T}$
will be provided in the proof of the following theorem that enumerates $%
DPL(n,4)$ codes for all $n\geq 2.$ We recall that two tilings $\mathcal{T=\{}%
V+l;l\in \mathcal{L}\}$ and $\mathcal{T}^{\prime }=\{V+l;l\in \mathcal{L}%
^{\prime }\}$ are called congruent (two codes $\mathcal{L}$ and $\mathcal{L}%
^{\prime }$ are called isomorphic) if there exists a linear
distance-preserving bijection $Z^{n}\rightarrow Z^{n}$ that maps $\mathcal{L}
$ on $\mathcal{L}^{\prime }.$\bigskip

\begin{theorem}
For each $n\geq 2,$ there are $2^{\aleph _{0}}$ non-isomorphic $DLP(n,4)$
codes.\bigskip
\end{theorem}

\begin{proof}
First of all we show that the total number of $DPL(n,4)$ codes is at most $%
2^{\aleph _{0}}.$ Each $DPL(n,4)$ code is a subset of $Z^{n}$ of cardinality 
$\aleph _{0}$. This in turn implies that the total number of $DLP(n,4)$
codes in $Z^{n}$ is at most the number of ways how to choose $\aleph _{0}$
points in $Z^{n}$; thus it is at most $(\aleph _{0})^{\aleph _{0}}=2^{\aleph
_{0}}.$ Therefore, there are at most $2^{\aleph _{0}}$ $DPL(n,4)$ codes.\ To
prove the statement we will construct $2^{\aleph _{0}}$ non-isomorphic $%
DPL(n,4)$ codes. Etzion \cite{E} proved that there are $2^{\aleph _{0}}$
non-isomorphic $DPL(n,4)$ codes for $n$ being a power of $2.$ So we are left
with the case when $n$ is not a power of $2.$\bigskip

\noindent Let $n=2^{t}(2k+1).\,\ $By Theorem \ref{F} there exists a
non-regular lattice tiling $\mathcal{T}$ $=\{W+l;l\in \mathcal{L\}}$ of $%
\mathbb{R}
^{n}$ by double-crosses, where $W$ is the $n$-double-cross comprising unit
cubes centered at words in $V=\{\pm e_{i},\pm e_{i}+e_{1};i=1,...,n\}$. For
the purpose of this proof it is convenient to know a basis of the lattice $%
\mathcal{L}.$ Let $j=2^{t}+1.$ Consider vectors $v_{i},i=1,...,n,$ given by $%
v_{1}=-\frac{1}{2}e_{1}+(2k+1)e_{j},$ $v_{j}=8ne_{j},$ and $v_{i}=\phi
(e_{i})e_{j}-e_{i},$ for $2\leq i\leq n,i\neq j$. It is easy to check that $%
\phi (v_{i})=0$ as $\phi (e_{j})=1$ and $\phi (\frac{1}{2}e_{1})=2k+1,$ see
the proof of Theorem \ref{F}, i.e., $v_{i}\in \mathcal{L}$ for all $%
i=1,...,n.$ Let $\,A$ be a matrix whose $i$-th row is the vector $v_{i}.$
Transpose the first and the $j$-th column and the first and the $j$-th row
of $A.$ An element of the resulting matrix $A^{\prime }$ is non-zero if and
only if it is in the first column or it is a diagonal element of $A^{\prime
} $. Thus $v_{i}$'s are independent vectors. Moreover, $\left\vert \det
(A^{\prime })\right\vert $ equals the product of its diagonal elements; thus 
$\left\vert \det (A)\right\vert =8n(-\frac{1}{2})(-1)^{n-2}=4n.$ Therefore $%
v_{i},i=1,...,n,$ constitute a basis of $\mathcal{L}.$\bigskip

\noindent Consider a relation $R$ on $\mathcal{L}$ given by $mRn$ if for
their first coordinates $m_{1},n_{1}$ it is $m_{1}-n_{1}\in Z.$ It is easy
to check that $R$ is an equivalence relation on $\mathcal{L}$ with two
equivalence classes $R_{1}=\{l;l\in L,l_{1}\in Z\}$ and $R_{2}=\{l;l\in 
\mathcal{L},l_{1}=k+0.5,k\in Z\}$ where $l_{1}$ is the first coordinate of $%
l.$ For $i=1,2,$ set $K_{i}=\bigcup\limits_{l\in \mathcal{R}_{i}}\{W+l\}.$
Clearly, $\mathcal{K}_{1}$ and $\mathcal{K}_{2}$ form a partition of $%
\mathbb{R}
^{n}.$ It is well known, see e.g. \cite{S}, \cite{Sz}, that $R_{1}$ is a
sublattice of $\mathcal{L}$ and that $R_{2}$ is a coset of $\mathcal{L}%
/R_{1};$ thus $R_{2}=R_{1}+v_{1}$ and also $\mathcal{K}_{2}=\mathcal{K}%
_{1}+v_{1}.$ Moreover, both $\mathcal{K}_{1}$ and $\mathcal{K}_{2}$
constitute prisms along $x_{1}$-axis. That is, if a point $x\in \mathcal{K}%
_{i},i\in \{1,2\}$ then the whole line parallel to $x_{1}$-axis passing
through $x$ belongs to $\mathcal{K}_{i}.$ As a very important consequence we
get that tiles in $\mathcal{K}_{2}$ can be shifted along $x_{1}$-axis
independently of tiles in $\mathcal{K}_{1}.$ For example, shifting $\mathcal{%
K}_{2}$ by $0.5e_{1}$ produces a tiling $\mathcal{S}$ of $%
\mathbb{R}
^{n}$ by $n$-double-crosses with all their cubes centered at integer points.
It is easy to show that $\mathcal{S}$ is not a lattice tiling but it is a
periodic tiling.\bigskip

\noindent Now we show that in fact both $\mathcal{K}_{1}$ and $\mathcal{K}%
_{2}$ consist of infinitely many connectivity (in the topological sense)
components. For the reader's convenience we start with $n=3.$ Then the basis
of $\mathcal{L}$ comprises vectors $v_{1}=-\frac{1}{2}%
e_{1}+3e_{2},v_{2}=24e_{2},$ and $v_{3}=13e_{2}-e_{3}.$ Note that cubes
centered at $\{O,\pm e_{2}\}$ belong to $\mathcal{K}_{1}$ as $V+O=V\in 
\mathcal{K}_{1}.$ Further, as $\mathcal{K}_{1}$ is a prism along $x_{1}$%
-axis, and $2tv_{1}\in R_{1}$ for all $t\in Z,$ we have that cubes centered
at $\{O,\pm e_{2}\}+6te_{2}\in \mathcal{K}_{1}$ for all $t\in Z.$ Similarly,
all cubes centered at $\{O,\pm e_{2}\}+e_{2}-e_{3}\in \mathcal{K}_{1}$ as $%
4v_{1}+v_{2}=-2e_{1}+e_{2}-e_{3}\in R_{1}.$ Thus all cubes centered at $%
\{O,\pm e_{2}\}+6te_{2}+s(e_{2}-e_{3})\in \mathcal{K}_{1}$ for all $s,t\in
Z; $ see Fig.~6, where the intersection of the $x_{2}x_{3}$-plane with $%
\mathcal{K}_{1}$ is shaded. The rest of the plane belongs to the
intersection with $\mathcal{K}_{2}$ as $\mathcal{K}_{2}=\mathcal{K}_{1}+v_{1}
$ and both $\mathcal{K}_{i}$s are prisms along $x_{1}$-axis. Thus we proved
that both $\mathcal{K}_{1}$ and $\mathcal{K}_{2}$ comprise infinitely many
components. In particular, if $B_{m},m\in Z,$ is the collection of all
components in $\mathcal{K}_{2}$ then the point $3(2m+1)e_{2}\in B_{m}$. The
proof for $n>3,$ $n$ not a power of $2$, would be very similar although
technically more involved. In this case the points $(2k+1)(2m+1)e_{2}$
belong to distinct connectivity components $B_{m}$ of $\mathcal{K}_{2},m\in
Z;$ we recall that $n=2^{t}(2k+1).$\bigskip %
\begin{figure}[t]
\centering
\epsfig{figure=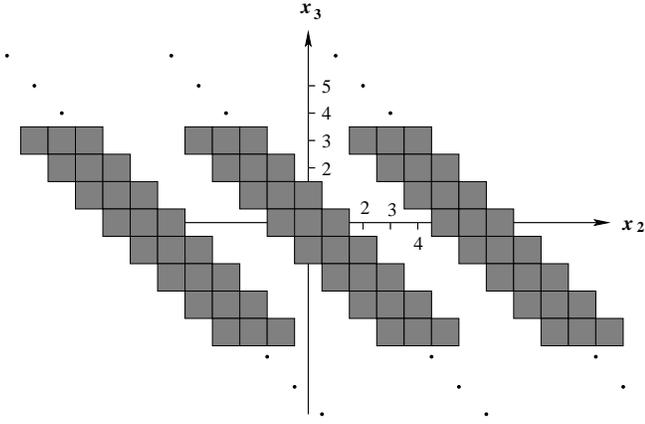, width=0.475\textwidth, angle=0}
\caption{The intersection of $x_{2}x_{3}$-plane with $K_{1}$}
\label{fig 6}
\end{figure}

\noindent The connectivity components of $\mathcal{K}_{2}$ can be shifted
along $x_{1}$-axis independently of each other and on $\mathcal{K}_{1}.$
This property of $\mathcal{K}_{2}$ will be used to construct $2^{\mathcal{%
\aleph }_{0}}$ non-isomorphic $DPL(n,4)$ codes. Let $r\in (0,1)$ be an
irrational number, $r=0.r_{1}...r_{k}...$ be its binary representation. A
tiling $\mathcal{T}_{r}$ will be obtained from the tiling $\mathcal{T}$ by
translating connectivity components $B_{m}$ of $\mathcal{K}_{2}$ as follows:
If $m\leq 0$ or $r_{m}=0$ then the component $B_{m}$ will be shifted by $-%
\frac{1}{2}e_{1},$ otherwise it will be shifted by $\frac{1}{2}e_{1}.$
Clearly all double-crosses in $\mathcal{T}_{r}$ are centered at integer
points. In other words, $\mathcal{T}_{r}=\{W+l;l\in \mathcal{L}^{\prime }\}$
is a tiling such that all coordinates of any $l\in \mathcal{L}^{\prime }$
are integers. Such tiling will be called a $Z$-tiling. Obviously, for any
two irrational numbers $r$ and $s,$ the corresponding tilings $\mathcal{T}%
_{r}$ and $\mathcal{T}_{s}$ are distinct; we recall that two tilings $%
\mathcal{\{}W\mathcal{+}l;l\in \mathcal{L}\}$ and $\mathcal{\{}W+l;l\in 
\mathcal{L}^{\prime }\}$ are called distinct if $\mathcal{L}$ and $\mathcal{L%
}^{\prime }$ are distinct. So we have produced a family $\mathcal{S=\{T}%
_{r}; $ $r$ is an irrational number in $(0,1)\}$ of $2^{\mathcal{\aleph }%
_{0}}$ distinct $Z$-tilings of $%
\mathbb{R}
^{n}$ by double-crosses. Each $Z$-tiling of $%
\mathbb{R}
^{n}$ by double-crosses induces a tiling of $Z^{n}$ by double-spheres. So
the family $\mathcal{S}$ induces a family $\mathcal{S}^{\prime }$ of $%
2^{\aleph _{0}}$ distinct tilings of $Z^{n}$ by double-spheres $V=\{\pm
e_{i},\pm e_{i}+e_{1};i=1,...,n\}$. Clearly $O\in \mathcal{L}$ for each
tiling $\mathcal{T=\{}V+l;l\mathcal{\in L}\},$ that is $V\in \mathcal{T}$
for each tiling $\mathcal{T}$ in $\mathcal{S}$. To each tiling $\mathcal{T}%
=\{V+l;l\in \mathcal{L}\}$ in $\mathcal{S}^{\prime }$ we assign a set $%
\mathcal{T}_{N}=:\{y;y$ is a center of $V+l,l\in \mathcal{L},$ and $%
\left\vert y\right\vert $ is even$\}=\{l;\left\vert l\right\vert $ is even$%
,l\in \mathcal{L}\}\cup \{l+e_{1};\left\vert l\right\vert $ is odd$,l\in 
\mathcal{L}\},$ where $\left\vert x\right\vert =\rho _{L}(x,O)$ is the Lee
weight of $x\in Z^{n}$. It was showed in the proof of Theorem \ref{X}, that $%
\mathcal{T}_{N}$ is a $DPL(n,4)$ code. It is easy to see that if $\mathcal{T}
$ and $\mathcal{T}^{\prime }$ are distinct tilings in $\mathcal{S}^{\prime }$
then the induced codes $\mathcal{T}_{N}$ and $\mathcal{T}_{N}^{\prime }$ are
distinct as well. Thus, we have constructed a family $\mathcal{S}^{\prime
\prime }$ of $2^{\aleph _{0}}$ distinct $DPL(n,4)$ codes. As each code in $%
\mathcal{S}^{\prime \prime }$ is isomorphic to at most $\aleph _{0}$ codes
in $\mathcal{S}^{\prime \prime }$(we recall that two codes $\mathcal{T}$ and 
$\mathcal{T}^{\prime }$are isomorphic if there exists a linear distance
preserving bijection mapping $\mathcal{T}$ onto $\mathcal{T}^{\prime }$) $%
\mathcal{S}^{\prime \prime }$ has to contain a set of $2^{\aleph _{0}}$
non-isomorphic $DPL(n,4)$ codes. The proof is complete.\bigskip
\end{proof}

\noindent As an easy consequence we get:\bigskip

\begin{corollary}
For each $n\geq 2,$ the number of non-periodic $DPL(n,4)$ codes is $%
2^{\aleph _{0}}.$\bigskip
\end{corollary}

\begin{proof}
It is easy to see that for each $p\in Z$ there are at most finitely many $p$%
-periodic $DPL(n,4)$ codes. Therefore, in aggregate, there are at most $%
\aleph _{0}$ periodic $DPL(n,4)$ codes. As there are in total $2^{\aleph
_{0}}$ $DPL(n,4)$ codes, $2^{\aleph _{0}}$ of them have to be
non-periodic.\bigskip
\end{proof}

\noindent We note that Szabo \cite{Sz} constructed non-regular lattice
tilings of $%
\mathbb{R}
^{n}$ by $n$-crosses if $2n+1$ is a prime. The non-existence of such lattice
tiling for $2n+1$ being a prime follows from a result of Redei \cite{Redei}
on factorization of Abelian groups. We have constructed a lattice
non-regular tiling of $R^{n}$ by $n$-double-cross for all $n$ that are not a
power of $2$. We believe that there are no non-regular tilings of $R^{n}$ by 
$n$-double-cross but at the moment we are not able to prove it.

\section{Decoding algorithm for linear perfect Lee codes}

\noindent In this section we design an algorithm for decoding a linear $%
DPL(n,d)$ code. The computational complexity of this algorithm is linear in $%
n$ for $d=3,4,$ and its complexity is $O(\log d)$ for $DPL(2,d)$ codes. We
recall that $DPL(n,d)$ codes are believed not to exists for $n\geq 3$ and $%
d>4.$ We stress that this algorithm can be applied to all perfect Lee codes
as diameter perfect Lee codes include also perfect error-correcting codes.
The efficiency of the algorithm is achieved thanks to the representation of
linear codes by means of a tiling constructed through a group homomorphism.
We design the algorithm for a linear $DPL(n,d)$ code. In case of a linear $%
DPL(n,d,q)$ code it suffices to realize that, by Theorem \ref{A}, the code
is a restriction of a $DPL(n,d)$ code and the corresponding algorithm would
be nearly identical.\bigskip

\noindent Let $\mathcal{L}$ be a linear $DPL(n,d)$ code. Theorem \ref{B}
implies that there is a homomorphism $\phi :Z^{n}\rightarrow (G,\circ ),$
where $G=Z^{n}/\mathcal{L}$ is an Abelian group of order $\left\vert
W\right\vert ,$ $W$ is the anticode of diameter $d-1$ of maximum size, and
that the restriction of $\phi ~$to $W$ is a bijection. Denote by $%
f:G\rightarrow W$ the function inverse to the restriction of $\phi $ to $W.$%
\bigskip

\noindent Assume that we receive a word $a\in Z^{n}.$ First we find the
element $g=\phi (a)$ of the group $G.$ $\ $If $g=e,$ the neutral element in $%
G,$ we are done as $a$ is a codeword in $\mathcal{L}$. In general the
corresponding codeword $l\in \mathcal{L}$ is given by $l=a-f(g)=a-f(\phi
(a)).$ Indeed, $\mathcal{T=\{}W+l;l\in \mathcal{L\}}$ is a tiling of $Z^{n}$
by $W.$ Thus, each element $a\in Z^{n}$ can be written in a unique way as $%
a=l+w,$ where $l\in \mathcal{L}$ is the sought codeword and $w\in W.$ Hence, 
$l=a-w=a-f(g)=a-f(\phi (a)).$

\bigskip

\noindent Hence the complexity of the algorithm is in fact the number of
operations needed for calculation of $f(\phi (a)).$ We note that the number
of operations in the calculation of $\phi (a)$ is linear in $n$ since $\phi
(a)=\phi ((a_{1},...,a_{n}))=\phi (e_{1})^{a_{1}}\circ ...\circ \phi
(e_{n})^{a_{n}}$, where the values $\phi (e_{i}),i=1,...,n,$ will be stored.
Now we bound from above the number of operations needed for the calculation
of the values of the function $f$ for any element $g$ in the group $G.$ The
values of the function $f$ need to be calculated only once and therefore
time needed for obtaining these values is not included in the complexity of
the decoding algorithm. The values of $f$ will be stored in a tabular way
and we will evaluate time needed to find the value of $f(g)$ in the table
for given $g\in G$. The size of the table is the order of the group $G,$
which equals $\left\vert W\right\vert .$ It is known, see \cite{E}, that the
volume of the double-sphere is 
\begin{equation}
\left\vert DS_{n,r}\right\vert =\sum\limits_{i=0}^{\min \{n-1,r\}}2^{i+1}{%
\binom{n-1}{i}}{\binom{r+1}{i+1}}
\end{equation}%
The volume of the sphere $S_{n,r}$ is smaller than the volume of $DS_{n,r}$
therefore we find the bound only for the double-sphere $DS_{n,r}$. However,
we can find the required value of $f$ in the table in an efficient way. As $%
G $ is an Abelian group, $G$ can be written as a direct product of cyclic
groups $Z_{q_{1}}\times ...\times Z_{q_{s}}$. We store values of $f$ in the
lexicographic order; that is, first comes value of $f$ at $(0,0,...,0,0),$
then value of $f$ \ at $(0,0,...,0,1),$ then at $(0,0,...,2),$ etc. The rank 
$r$ of an element $a=(a_{1},...,a_{s})$ in this order is 
\begin{equation}
r(a)=1+a_{s}+\sum\limits_{i=1}^{s-1}a_{i}\prod\limits_{j=i+1}^{s}q_{j}
\end{equation}

\noindent Hence, the rank $r(a)$ can be calculated in time linear with $s$
as from (2) it follows 
\begin{equation}
\begin{split}
r(a)=& 1+a_{s}+ \\
& q_{s}(a_{s-1}+q_{s-1}(a_{s-2}+q_{s-2}(a_{s-3}+...+ \\
& q_{3}(a_{2}+q_{2}(a_{1}))...))) \\
&
\end{split}%
\end{equation}%
%
%

\noindent The maximum value of $s$ is achieved when $G=Z_{2}\times
Z_{2}\times ...\times Z_{2}.$ Thus the maximum value of $s=\log \left\vert
W\right\vert =\log \left\vert DS_{n,r}\right\vert $. For $n=2$ we get $%
\left\vert DS_{2,r}\right\vert =2(r+1)^{2},$ thus $s=O(\log d)$ where $%
d=2r+2.$ For $d=4$ it is $\left\vert DS_{n,1}\right\vert =4n,$ thus $s=\log
n.$ In aggregate, the complexity of the algorithm for $DPL(n,d),d=3,4,$
codes is $\Theta (n),$ and its complexity is $O(\log d)$ for $DPL(2,d)$
codes.

\end{document}